\documentclass[11pt]{article}
\usepackage[utf8]{inputenc}
\usepackage[a4paper]{geometry}
\usepackage{comment}

\usepackage{Style/en-math}
\usepackage{Style/preamble}
\usepackage{Style/numerical-semigroups}
\usepackage{tikz-cd}

\title{Distributed matrix multiplication with straggler tolerance using algebraic function fields }
\author{
  Fidalgo-D{\'i}az, Adri{\'a}n\\
  \texttt{adrian.fidalgo22@uva.es}
  \and
  Mart{\'i}nez-Pe{\~n}as, Umberto\\
  \texttt{umberto.martinez@uva.es}
}
\date{}

\begin{document}

\maketitle

{\let\thefootnote\relax\footnotetext{The authors are with the IMUVa-Mathematics Research Institute, University of Valladolid, Spain. \\The authors gratefully acknowledge the support from a Mar{\'i}a Zambrano contract by the University of Valladolid, Spain (Contract no. E-47-2022-0001486), and the support from MCIN/AEI/10.13039/501100011033 and the European Union NextGenerationEU/PRTR (Grant no. TED2021-130358B-I00).}}

\begin{abstract}
The problem of straggler mitigation in distributed matrix multiplication (DMM) is considered for a large number of worker nodes and a fixed small finite field. Polynomial codes and matdot codes are generalized by making use of algebraic function fields (i.e., algebraic functions over an algebraic curve) over a finite field. The construction of optimal solutions is translated to a combinatorial problem on the Weierstrass semigroups of the corresponding algebraic curves. Optimal or almost optimal solutions are provided. These have the same computational complexity per worker as classical polynomial and matdot codes, and their recovery thresholds are almost optimal in the asymptotic regime (growing number of workers and a fixed finite field). 


\end{abstract}



\section{Introduction}
Since the saturation of Moore's Law became a reality, parallel algorithms emerged as a solution to continue speeding up computations, reaching the point where making use of these techniques of distributed computation is a standard practice when designing a high performance architecture. The main idea of these algorithms consists on splitting up the initial problem into smaller ones that can be solved simultaneously, say by the nodes of a network. Then all the results are gathered and the solution to the original problem is given.

As this model scales by using more nodes, the straggling effect begins to ocurr. As the execution time of these algorithms is limited by the execution time of the slowest of the nodes, this induces a bottleneck, since the larger the number of nodes, the larger the expected difference between their execution times. It is necessary to mitigate this effect to avoid an inefficient implementation of these architectures when the amount of nodes grows.

In this manuscript, we study straggler mitigation in Distributed Matrix Multiplication, from now on DMM. The idea of DMM is multiplying two matrices $A$ and $B$ by considering them as block matrices and computing each multiplication of blocks separately. Matrix multiplication is of great importance since it is at the core of machine learning algorithms and signal processing.

\begin{figure}[h]
    \centering
    \includegraphics[width=8cm]{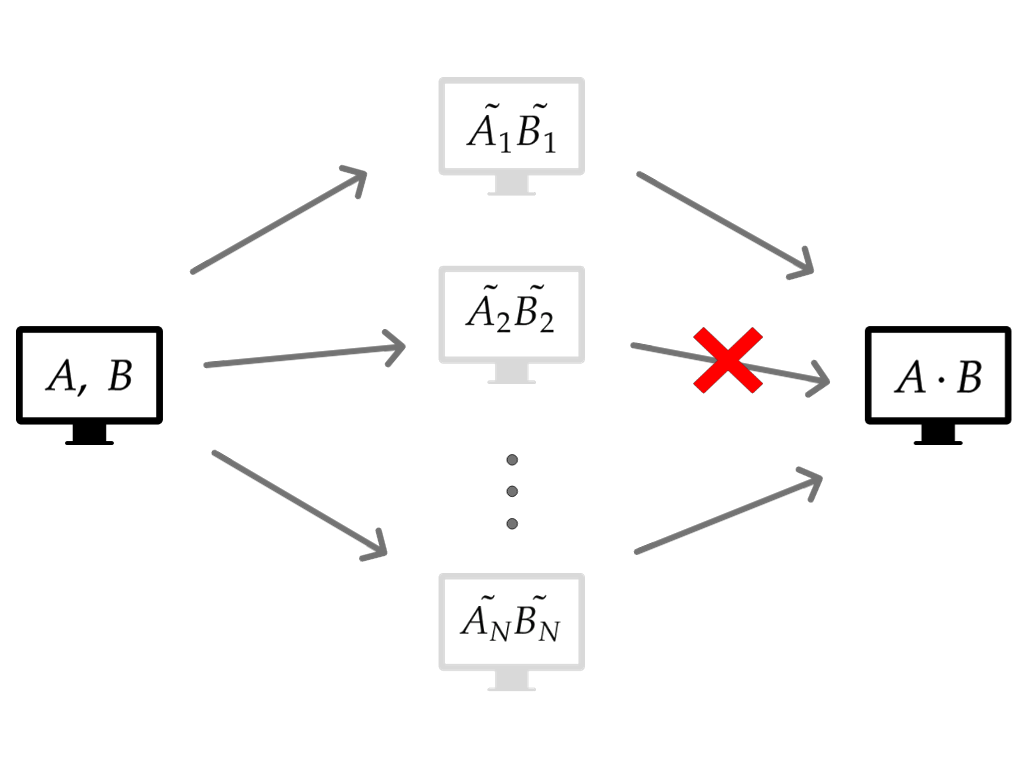}
    \caption{General scheme of DMM with straggler tolerance, where each $\tilde{A}_i$ and $\tilde{B}_i$ denotes a matrix of lower size than $A$ and $B$.}
\end{figure}

In the recent years, new ideas concerning DMM and, in general, coded computation have emerged. The initial approach was to achieve straggler mitigation by encoding one of the two operands of the operation to compute \cite{dutta2017coded}. This was improved in \cite{polynomial_codes} by employing evaluations of polynomials as Reed-Solomon codes do. Following this line, new methods came up \cite{matdot_codes, polydot_entanglement} that exploit new ways of encoding matrices in ways that allow recovering the product from fewer worker nodes. More recently, certain algebraic function fields have been proposed for secure DMM \cite{machado2023hera, makkonen2023algebraic}, extending the use of polynomials for this problem \cite{GASP_scheme}.

In this work, we generalize polynomial codes \cite{polynomial_codes} and matdot codes \cite{matdot_codes} for straggler mitigation in DMM by using algebraic function fields. While polynomial codes are optimal in terms of communication cost and computation per node, matdot codes are optimal in terms of the number of nodes needed to recover the original product (recovery threshold). However, in both cases the number of workers is upper bounded by the size of the base field of the matrices. Our generalization uses algebraic function fields to solve this problem, allowing us to define algorithms that use a large number of nodes while being able to work over small fields. This fits in well with the tendency in computation of increasing the number of nodes in parallelization to obtain faster algorithms.




\section{Preliminaries} \label{section:preliminaries}
We now introduce some preliminaries on algebraic function fields and numerical semigroups. For more details on the former, we refer the reader to \cite{stichtenoth}.

\begin{definition}
    Let $F/\mathbb{F}_q$ be an algebraic function field of trascendence degree $1$ over $\mathbb{F}_q$, the finite field with $q$ elements. Let $\mathbb{P}$ the set of places of $F/\mathbb{F}_q$. We consider a divisor $D$, i.e., a formal sum $\sum_{P \in \mathbb{P}} n_P P$ where $n_P = 0$ for almost all places. Let $f \in F$, we define the divisor
    \begin{equation*}
        (f) := \sum_{P \in \mathbb{P}} \nu_P (f) P,
    \end{equation*}
    where $ \nu_P $ denotes the valuation at $ P $. 
    For a divisor $D$, we say that $D \geq 0$ if $n_P \geq 0$ for each $P \in \mathbb{P}$. We define the Riemann-Roch space of a divisor $D$ as
    \begin{equation*}
        \mathcal{L}(D) := \{ f \in F^* \such (f) + D \geq 0\} \cup \{ 0 \}.
    \end{equation*}
\end{definition}

Let $P_1, P_2, \ldots, P_N, Q \in \mathbb{P}$ be $N+1$ distinct $ \mathbb{F}_q $-rational places\footnote{$Q$ does not need to be $ \mathbb{F}_q $-rational, but in general Weierstrass semigroups are computed over $\mathbb{F}_q$-rational places, hence we assume that $ Q $ is $ \mathbb{F}_q $-rational.}. Consider the evaluation map given by
\begin{equation*}
    \begin{split}
        ev: \mathcal{L}(kQ) &\to \mathbb{F}_q^N\\
        f &\mapsto (f(P_1), f(P_2), \ldots, f(P_N)),
    \end{split}
\end{equation*}
for $k \in \mathbb{N}$. Observe that the kernel of the evaluation map is $\ker (ev) = \mathcal{L}(kQ - \sum_{i=1}^N P_i)$. If $k < N$, then the evaluation map is injective by \cite{stichtenoth, Lemma 1.4.7} and we can identify each algebraic function $f \in \mathcal{L}(kQ)$ with its evaluations in $P_1, P_2, \ldots, P_N$. Moreover, the function may be algorithmically retrieved by performing interpolation when enough evaluations are known (see Section \ref{section:interpolating}).

\begin{notation}
    We write $\mathbb{N} := \mathbb{Z}_{\geq 0}$.
\end{notation}

\begin{definition}
    Let $Q \in \mathbb{P}$, we define
    \begin{equation*}
        \mathcal{L}(\infty Q) := \bigcup_{k \in \mathbb{N}} \mathcal{L}(kQ),
    \end{equation*}
    the set of algebraic functions that do not have poles different than $Q$. Observe that $\mathcal{L}(\infty Q)$ is a $\mathbb{F}_q$-subalgebra of $F$.
\end{definition}

\begin{definition}
    We define the Weierstrass semigroup of $Q$ as
    \begin{equation*}
        S = \{k \in \mathbb{N} \such \mathcal{L}(kQ) \neq \mathcal{L}((k-1)Q)\}.
    \end{equation*}
\end{definition}

The Weierstrass semigroup is indeed a numerical semigroup (see Definition \ref{definition:numerical_semigroup}) and its number of gaps (see Definition \ref{definition:gaps}) is exactly the genus of $F/\mathbb{F}_q$. In addition, $\mathcal{L}(\infty Q)$ is generated by $\{f \in F \such \nu_Q(f) \in S\}$ as a $\mathbb{F}_q$-vector space and so
\begin{equation*}
    \mathcal{L}(\infty Q) = \bigcup_{s \in S} \mathcal{L}(sQ).
\end{equation*}

Now we focus on numerical semigroups. See \cite{numerical_semigroups} for more details.

\begin{definition} \label{definition:numerical_semigroup}
    We say that the set $S \subseteq \mathbb{N}$ is a numerical semigroup if
    \begin{itemize}
        \item $(S, +)$ is a monoid, i.e., $0 \in S$ and $s_1 + s_2 \in S$ for every $s_1, s_2 \in S$.
        \item $\mathbb{N} \setminus{S}$ is finite.
    \end{itemize}
\end{definition}

\begin{notation}
    Let $A \subseteq \mathbb{N}$, we write $A^* := A \setminus{\{0\}}$. For integers $ m \leq n $, we denote $ [m,n] = \{ m,m+1,\ldots,n \} $ and $ [m,\infty) = \{ m,m+1,m+2,\ldots \} $.
\end{notation}

\begin{definition} \label{definition:gaps}
    Let $S$ be a numerical semigroup.
    \begin{itemize}
        \item We define its conductor as $ c(S) = \min \{ s \in S : [s,\infty) \subseteq S \} $.
        \item We define $ n(S) = |S \cap [0,c(S)-1]| $ and $ g(S) = |\mathbb{N} \setminus S| $.
        \item We define the multiplicity of $ S $ as $ m(S) = \min (S^*) $.
    \end{itemize}
\end{definition}

\section{Algebraic-geometry polynomial codes} \label{section:polynomial}

\subsection{The classical solution} \label{subsection:polynomial}

Consider $A \in \mathbb{F}_q^{r \times s}$ and $B \in \mathbb{F}_q^{s \times t}$, two matrices of sizes $r \times s$ and $s \times t$, respectively, over $\mathbb{F}_q$. We want to compute the product $A B$ in a way that is both parallelizable and straggler resistant, i.e., we do not need the output of all the worker nodes to recover the product. The following method for achieving this is called polynomial codes and was proposed in \cite{polynomial_codes}. In comparison with matdot codes, presented in Section \ref{section:matdot}, polynomial codes obtain better communication and computation costs per node but worse recovery threshold.

Start by splitting up the matrices into submatrices
\begin{equation*}
    \begin{split}
    A = \begin{pmatrix}
        A_1\\
        A_2\\
        \vdots\\
        A_m
        \end{pmatrix}
    ,\quad
    B = \begin{pmatrix}
        B_1 & B_2 & \ldots & B_n
        \end{pmatrix},
    \end{split}
\end{equation*}
with $A_i \in \mathbb{F}_q^{\frac{r}{m} \times s}$ and $B_i \in \mathbb{F}_q^{s \times \frac{t}{n}}$. Observe that, with this subdivision,
\begin{equation*}
\begin{split}
    A B =
    \begin{pmatrix}
        A_1 B_1 & A_1 B_2 & \ldots & A_1 B_n\\
        A_2 B_1 & A_2 B_2 & \ldots & A_2 B_n\\
        \vdots & \vdots & \ddots & \vdots\\
        A_m B_1 & A_m B_2 & \ldots & A_m B_n\\
    \end{pmatrix}.
\end{split}
\end{equation*}
Choose two polynomials (with matrix coefficients)
\begin{equation*}
\begin{split}
    p_A(x) := \sum_{i=1}^m A_i x^{a_i}
    ,\quad
    p_B(x) := \sum_{i=1}^n B_i x^{b_i},
\end{split}
\end{equation*}
in such way that
\begin{equation} \label{equation:polynomial_codes}
\begin{split}
    a_i + b_j \neq a_k + b_l \quad \text{if} \quad (i,j) \neq (k,l).
\end{split}
\end{equation}
Because the degrees of the monomials of $p_A$ and $p_B$ satisfy (\ref{equation:polynomial_codes}), every submatrix $A_i B_j$ present in the product $A B$ is recoverable as the coefficient of the monomial of degree $a_i + b_j$ of the polynomial $h := p_A p_B$. In a certain way, $p_A$, $p_B$ and $h$ ``encode the information'' about the matrices $A$, $B$ and $AB$, respectively.

The following algorithm computes $A B$ by using $N$ worker nodes. First, the master node chooses $N$ distinct points $x_1, x_2, \ldots, x_N \in \mathbb{F}_q$ and shares one of them with each worker node together with the polynomials $p_A$ and $p_B$. Then, the nodes compute the evaluation of $h$ in the corresponding point, say $x_i$, by using $h(x_i) = p_A(x_i) p_B(x_i)$. After a sufficient number of nodes are done with the computation of $p_A(x_i) p_B(x_i)$, enough to interpolate $h$, the master node uses the evaluations available for recovering $h$ via interpolation and, consequently, recovers $A B$. Figure \ref{figure:dmm_framework} summarizes the algorithm.

The straggler resistance comes from the fact that we only need some evaluations to recover $h$. The minimum number of workers that need to finish their executions to recover the product $A B$ (the so-called recovery threshold) is the number of evaluations needed to perform Lagrange interpolation, which is
$$ \deg(h)+1 = \deg(p_A) + \deg(p_B)+1 = \max(D_A) + \max(D_B)+1, $$
where $ D_A = \{ a_1, a_2 \ldots, a_m \} $ and $ D_B = \{ b_1, b_2, \ldots, b_n \} $. Construction \ref{construction:optimal_polynomial} shows how to choose properly the degrees of the monomials present in $p_A$ and $p_B$.

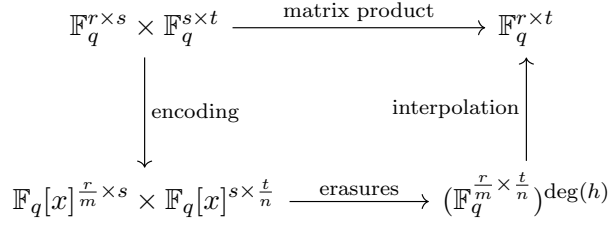
\begin{figure}
    \centering
    \begin{tikzcd}[sep=huge]
        \mathbb{F}_q^{r \times s} \times \mathbb{F}_q^{s \times t} \arrow[r, "\text{matrix product}"] \arrow[d, "\text{encoding}"] &
        \mathbb{F}_q^{r \times t} &\\
        \mathbb{F}_q[x]^{\frac{r}{m} \times s} \times  \mathbb{F}_q[x]^{s \times \frac{t}{n}} \arrow[r, "\text{erasures}"] &
        (\mathbb{F}_q^{\frac{r}{m} \times \frac{t}{n}})^{\deg(h)} \arrow[u, "\text{interpolation}"] &
    \end{tikzcd}
    \caption{Commutative diagram summarizing the polynomial code algorithm.}
    \label{figure:dmm_framework}
\end{figure}

\begin{notation}
    Let $A, B \subseteq \mathbb{N}$, we write $A + B$ to denote the Minkowski sum, i.e.,
    \begin{equation*}
        A + B := \{a + b \in \mathbb{N} \such a \in A, \quad b \in B\}.
    \end{equation*}
\end{notation}

\begin{construction}[\textbf{\cite{polynomial_codes}}] \label{construction:optimal_polynomial}
    Define the sets
    \begin{equation*}
        \begin{split}
        &D_A := \{0, 1, \ldots, m-1\},\\
        &D_B := \{0, m, \ldots, (n-1) m\}.
        \end{split}
    \end{equation*}
    These sets satisfy (\ref{equation:polynomial_codes}) so they are a valid option to define the polynomials $p_A$ and $p_B$. These sets are ``optimal'' in the sense that $h$ has the least possible degree ($\deg(h)+1 = \max(D_A) + \max(D_B)+1 = mn $\footnote{If $ A_m \neq 0 $, $ B_n \neq 0 $ and $ A_mB_n = 0 $, then $\deg(h) < mn - 1$, but this is a very rare scenario.}). This optimality is not difficult to proof since for any sets $D_A^\prime, D_B^\prime \subseteq \mathbb{N}$ of sizes $m$ and $n$ that satisfy (\ref{equation:polynomial_codes}) (equivalently $|D_A^\prime + D_B^\prime| = |D_A^\prime| |D_B^\prime|$) we have
    \begin{equation*}
        \max(D_A^\prime) + \max(D_B^\prime) +1 = \max(D_A^\prime + D_B^\prime) +1 \geq |D_A^\prime + D_B^\prime|  = mn .
    \end{equation*}
    A more general proof in terms of information theory can be found in \cite{polynomial_codes, Theorem 1}.
\end{construction}
%

\subsection{The Algebraic-Geometry solution} \label{subsection:ag_polynomial}

The main drawback of polynomial codes is that, since we need to choose distinct $ x_1, x_2, \ldots, x_N \in \mathbb{F}_q $, the number of workers must satisfy $ N \leq q $. We may circumvent this issue by replacing polynomials with algebraic functions and choosing points on an algebraic curve. The increase in computational complexity is negligible (see Section \ref{section:interpolating}) and the loss of optimality in the parameters is asymptotically small (see Section \ref{section:asymptotic}).

In polynomials codes, the product $AB$ is recovered as the coefficients of $h$ with respect to the basis formed by the monomials $x^i$. As an $\mathbb{F}_q$-algebra, the set of polynomials $\mathbb{F}_q[x]$ is generalized to $\mathcal{L}(\infty Q)$, so the polynomials $p_A$, $p_B$ and $h$ become algebraic functions.

Let $F/\mathbb{F}_q$ be an algebraic function field and $Q \in \mathbb{P}$ a place. Fix a basis of $\mathcal{L}(\infty Q)$
\begin{equation*}
    \begin{split}
        \mathcal{B} := \{f_{s_0}, f_{s_1}, \ldots\},
    \end{split}
\end{equation*}
with $f_{s_i} \in F$ and $\nu(f_{s_i}) = s_i$ the $(i+1)$th element of the Weierstrass semigroup of $Q$, denoted by $S$. Let us generalize (\ref{equation:polynomial_codes}):

\begin{definition} \label{definition:AG_poly_codes_solution}
    Given $m, n \in \mathbb{N}^*$, we say that $D_A, D_B \subseteq S$ is a solution to the AG polynomial code problem if
    \begin{itemize}
        \item The sets $D_A$ and $D_B$ are of size $m$ and $n$, respectively.
        \item $a + b \neq a^\prime + b^\prime$ for every $(a, b), (a^\prime, b^\prime) \in D_A \times D_B$ such that $(a, b) \neq (a^\prime, b^\prime)$.
    \end{itemize}
    We define the recovery threshold as $\max(D_A) + \max(D_B)+1$ and we say that the solution is optimal if the recovery threshold is minimum among all the possible solutions.
\end{definition}

\begin{remark}
    Observe that if $S = \mathbb{N}$, Definition \ref{definition:AG_poly_codes_solution} recovers (\ref{equation:polynomial_codes}).
\end{remark}

In order to proceed with the DMM algorithm, we need to tweak the basis $\mathcal{B}$. Consider a solution $D_A$ and $D_B$ to the AG polynomial code problem. Order the set $D_A \times D_B$ as
\begin{equation*}
    \begin{split}
    (a, b) \leq (a^\prime, b^\prime) \iff a + b \leq a^\prime + b^\prime.
    \end{split}
\end{equation*}
Then for each $(a, b) \in D_A \times D_B$, in ascending order, we set
\begin{equation*}
    \begin{split}
    f_{a + b} := f_a f_b.
    \end{split}
\end{equation*}

\begin{notation}
    Given $g = \sum_{s \in S} \lambda_s f_s$, we denote $\pi_{f_s}(g) = \lambda_s$, the projection of $g$ over $f_s$.
\end{notation}

Then, we define the functions $p_A \in \mathcal{L}(\infty Q)^{\frac{r}{m} \times s}$ and $p_B \in \mathcal{L}(\infty Q)^{s \times \frac{t}{n}}$ as
\begin{equation*}
    \begin{split}
    p_A := \sum_{i=1}^m A_i f_{a_i}
    ,\quad
    p_B := \sum_{i=1}^n B_i f_{b_i},
    \end{split}
\end{equation*}
together with $h := p_A p_B$. We proceed similarly as in the polynomial code scheme. Let $P_1, P_2, \ldots, P_N \in \mathbb{P}$ be $ \mathbb{F}_q $-rational places. The master node shares $p_A$, $p_B$ and a place $P_i$ with each worker. Each of them computes $h(P_i) = p_A(P_i) p_B(P_i)$ using its corresponding place. When enough workers have finished, at least $v_Q(h)+1 = \max(D_A) + \max(D_B)+1$ (we will cover the details in Section \ref{section:interpolating}), the master node gathers the evaluations and interpolates $h$, recovering each of the products $A_i B_j$ as the projection $\pi_{f_{a_i + b_j}} (h)$. Observe that this algorithm is the result of substituting $\mathbb{F}_q[x]$ for $\mathcal{L}(\infty Q)$ in Figure \ref{figure:dmm_framework}.

\begin{remark} \label{remark:ordering_DADB}
The purpose of ordering $D_A \times D_B$ is to be able to define $ f_{a+b} = f_af_b $ consistently for $ (a,b) \in D_A \times D_B $. Notice that it may happen that $a + b \in D_A$, and by ordering $D_A \times D_B$ as above, the definition of $f _{(a + b) + b^\prime} := f_{a+b} f_{b^\prime} = f_af_bf_{b^\prime}$, for $b^\prime \in D_B$, is consistent. The cost of ordering $D_A \times D_B$ can vary drastically depending on the sets $D_A$ and $D_B$. However, if $0 \in D_A \cap D_B$, ordering $D_A \times D_B$ is not needed since, in that case, we may choose $ f_0 := 1 $ and it holds $ a+b \notin D_A \cup D_B $ for any $ a \in D_A^* $ and $ b \in D_B^* $ by Definition \ref{definition:AG_poly_codes_solution}. 
%
\end{remark}

The only missing part for completing the algorithm is the explicit construction of the sets $D_A$ and $D_B$. We dedicate the remainder of the section to finding solutions to the AG polynomial code problem and giving bounds on the recovery threshold.

\begin{construction} \label{construction:AG_polynomial_trivial}
    Define the sets
    \begin{equation*}
        \begin{split}
            &D_A := c(S) + \{0, 1, \ldots, m - 1\} = \{c(S), c(S) + 1, \ldots, c(S) + m - 1\},\\
            &D_B := c(S) + \{0, m, \ldots, (n - 1) m\} = \{c(S), c(S) + m, \ldots, c(S) + (n - 1 ) m\}.
        \end{split}
    \end{equation*}
    This is a trivial solution to the AG polynomial code problem constructed from the optimal solution when $S = \mathbb{N}$, given in Construction \ref{construction:optimal_polynomial}. It yields a recovery threshold of $2 c(S) + m n $.
\end{construction}

For our next construction, we need the following definition and lemma from \cite{numerical_semigroups}.

\begin{definition}
    Let $S$ be a numerical semigroup and $n \in S^*$. We define the Ap{\'e}ry set with respect to $n$ as
    \begin{equation*}
        \Ap(S, n) := \{s \in S \such s - n \notin S\}.
    \end{equation*}
\end{definition}

\begin{lemma} \label{lemma:Kunz_coord}
    Let $S$ be a numerical semigroup and $n \in S^*$. Then
    \begin{itemize}
        \item $\Ap(S, n) = \{w_0, w_1, \ldots, w_{n-1}\}$, where $w_i$ is the lowest element of $S$ congruent with $i \mod n$.
        \item $c(S) = \max(\Ap(S, n)) - n + 1$.
    \end{itemize}
\end{lemma}

\begin{construction} \label{construction:AG_polynomial_Apery}
    Let $ m^\prime  := \min \{s \in S \such s \geq m\}$. Choose $D_A$ as the subset formed by the first $m$ elements of the Ap{\'e}ry set $\Ap(S, m^\prime )$. Define $D_B$ as
    \begin{equation*}
        \begin{split}
        D_B := \{0, m^\prime , \ldots, (n-1) m^\prime \}.
        \end{split}
    \end{equation*}
    This is a solution for the AG polynomial code problem by Lemma \ref{lemma:Kunz_coord}. Its recovery threshold satisfies the following upper bound
    \begin{equation*}
        \begin{split}
            \max(D_A) + \max(D_B) +1 &= \max(D_A) + m^\prime  (n - 1) +1 \\
            &\leq \max(\Ap(S, m^\prime )) + m^\prime  (n - 1) +1\\
            &= c(S) + m^\prime  - 1 + m^\prime  (n - 1) +1\\
            &= c(S) + m^\prime  n\\
            &\leq c(S) + (m + m(S) - 1) n .
        \end{split}
    \end{equation*}
    Observe that if $m \in S$, then $ m^\prime  = m$ and
    \begin{equation*}
        \begin{split}
        \max(D_A) + \max(D_B) +1 = c(S)  + m + (n-1)m = c(S) + m n .
        \end{split}
    \end{equation*}
    Moreover, observe that both $D_A$ and $D_B$ contain $0$, so we do not need to order the set $D_A \times D_B$ for modifying the basis $\mathcal{B}$ by Remark \ref{remark:ordering_DADB}.
\end{construction}

For our next construction, we make use of the following lemma.

\begin{lemma} \label{lemma:AG_polynomial}
    Let $D_A, D_B \subseteq S$. Consider the sets
    \begin{equation*}
        \begin{split}
        &E_A := \{d - d^\prime  \in \mathbb{N} \such d, d^\prime  \in D_A, \quad d > d^\prime \},\\
        &E_B := \{d - d^\prime  \in \mathbb{N} \such d, d^\prime  \in D_B, \quad d > d^\prime \}.
        \end{split}
    \end{equation*}
    The sets $D_A$ and $D_B$ are a solution to the AG polynomial code problem if and only if $E_A \cap E_B = \emptyset$.
\end{lemma}

\begin{proof}
    There exist distinct elements $(d_A, d_B), (d_A^\prime , d_B^\prime ) \in D_A \times D_B$ such that $d_A + d_B = d_A^\prime  + d_B^\prime $ if and only if there are $e_A := d_A - d_A^\prime  \in E_A$ and $e_B := d_B - d_B^\prime  \in E_B$ such that $e_A = e_B$.
\end{proof}

\begin{construction} \label{construction:AG_polynomial_lemma}
    Consider the sets
    \begin{equation*}
        \begin{split}
            &D_A := \{c(S), c(S) + 1, \ldots, c(S) + m - 1\},\\
            &D_B := \{m_1, m_2, \ldots m_n\},
        \end{split}
    \end{equation*}
    where each $m_i$ is defined recursively as
    \begin{equation*}
        \begin{split}
        m_i :=
        \begin{cases}
            0 &\text{ if } i = 1,\\
            \min \{ s \in S \such s \geq m_{i-1} + m\} &\text{ if } i > 1.
        \end{cases}
        \end{split}
    \end{equation*}
    Applying Lemma \ref{lemma:AG_polynomial}, we conclude that $D_A$ and $D_B$ form a solution to the AG polynomial code problem, since
    \begin{equation*}
        \begin{split}
            &\max(E_A) = m - 1,\\
            &\min(E_B) \geq m,
        \end{split}
    \end{equation*}
    which implies that $E_A \cap E_B = \emptyset$. Let us study the recovery threshold of this solution. First, we have that $\max(D_A) = c(S) + m - 1$. Second, define $\mu_i := m_{i+1} - m_i - m$ for each $i = 1, 2, \ldots, n - 1$, and we observe that
    \begin{equation*}
        \begin{split}
        \max(D_B) = \sum_{i=1}^{n - 1} (m_{i+1} - m_i) = \sum_{i=1}^{n - 1} (\mu_i + m) = m (n - 1) + \sum_{i=1}^{n-1} \mu_i.
        \end{split}
    \end{equation*}
    Combining both equalities,
    \begin{equation*}
        \begin{split}
        \max(D_A) + \max(D_B) +1 = c(S) + m n + \sum_{i=1}^{n-1} \mu_i.
        \end{split}
    \end{equation*}
    Observe that if $m \in S$, then $\mu_i = 0$ for each $i = 1, 2, \ldots, n-1$ and
    \begin{equation*}
        \begin{split}
        \max(D_A) + \max(D_B) +1 = c(S) + m n .
        \end{split}
    \end{equation*}
\end{construction}

\begin{remark}
    Consider the sets $D_A$ and $D_B$ from Construction \ref{construction:AG_polynomial_lemma}. Define
    \begin{equation*}
        D_A^\prime  := (D_A \setminus \{s\}) \cup \{0\},
    \end{equation*}
    where $s$ is the unique element of the semigroup such that $c(S) \leq s \leq c(S) + m-1$ and $s$ is divisible by $m$. Suppose that $m \in S$. Then $E_B$ consists of multiples of $m$ and the sets $D_A^\prime $ and $D_B$ are a solution to the AG polynomial code problem, i.e., $E_A^\prime  \cap E_B = \emptyset$. This solution is useful due to the fact that $0 \in D_A \cap D_B$ and Remark \ref{remark:ordering_DADB}.
\end{remark}

Table \ref{table:AG_polynomial} summarizes the constructions solving the AG polynomial code problem.

\begin{table}[h]
    \centering
    \begin{tabular}{l|l|lll}
        & \textbf{if} $m \notin S$ & \textbf{if} $m \in S$ &  &  \\ \cline{1-3}
        \multicolumn{1}{l|}{\bfseries Construction \ref{construction:AG_polynomial_trivial}} & $2 c(S) + m n $                       & $2 c(S) + m n $ &  &  \\ \cline{1-3}
        \multicolumn{1}{l|}{\bfseries Construction \ref{construction:AG_polynomial_Apery}}   & $c(S) + m^\prime  n $                        & $c(S) + m n $   &  &  \\ \cline{1-3}
        \multicolumn{1}{l|}{\bfseries Construction \ref{construction:AG_polynomial_lemma}}   & $c(S) + mn + \sum_{i=1}^{n-1} \mu_i$ & $c(S) + m n $   &  &  \\ \cline{1-3}
    \end{tabular}
    \caption{Recovery threshold of proposed solutions to the AG polynomial code problem.}
    \label{table:AG_polynomial}
\end{table}

We conclude this section with a lower bound on the recovery threshold, which will be used in Section \ref{section:asymptotic} to show that Constructions \ref{construction:AG_polynomial_Apery} and \ref{construction:AG_polynomial_lemma} are asymptotically optimal.
\begin{proposition} \label{proposition:AG_polynomial_bound}
    Let $D_A$ and $D_B$ be a solution to the AG polynomial code problem. If $m n \geq n(S)$, then
    \begin{equation*}
        \begin{split}
        \max(D_A) + \max(D_B) +1 \geq g(S) + m n .
        \end{split}
    \end{equation*}
\end{proposition}

\begin{proof}
    Consider
    \begin{equation*}
        \begin{split}
        & g^\prime  := |\{x \in \mathbb{N}\setminus{S} \such x \leq \max(D_A + D_B)\}|,\\
        & n^\prime  := |\{s \in S \such s \leq \max(D_A + D_B)\}|,
        \end{split}
    \end{equation*}
    and observe that $\max(D_A + D_B) = g^\prime  + n^\prime  - 1$. If $m n \geq n(S)$, then $\max(D_A + D_B) \geq c(S)$ since
    \begin{equation} \label{equation:AG_polynomial_bound_proof}
        D_A + D_B \subseteq S \text{ and } |D_A + D_B| = m n
    \end{equation}
    We deduce that $\max(D_A) + \max(D_B) \geq c(S)$ and $ g^\prime  = g(S)$. In addition, $ n^\prime  \geq mn$ because of the definition of $ n^\prime  $ and (\ref{equation:AG_polynomial_bound_proof}). Using these inequalities we obtain
    \begin{equation*}
        \begin{split}
        \max(D_A + D_B) +1 = g(S) + n^\prime  \geq g(S) + m n.
        \end{split}
    \end{equation*}
\end{proof}

\section{Algebraic-geometry matdot codes} \label{section:matdot}

\subsection{The classical solution} \label{subsection:matdot}

We now consider the matdot codes from \cite{matdot_codes}, which split $A$ and $B$ in a different way. In contrast with polynomial codes (Section \ref{section:polynomial}), matdot codes offer optimal recovery threshold at the cost of performing worse in communication and computation. Let
\begin{equation*}
    \begin{split}
    A = \begin{pmatrix}
        A_1 & A_2 & \ldots & A_m
        \end{pmatrix}
    ,\quad
    B = \begin{pmatrix}
        B_1 \\
        B_2 \\
        \vdots\\
        B_m
        \end{pmatrix},
    \end{split}
\end{equation*}
with $A_i \in \mathbb{F}_q^{r \times \frac{s}{m}}$ and $B_i \in \mathbb{F}_q^{\frac{s}{m} \times t}$. This subdivision yields
\begin{equation*}
    \begin{split}
    A B = \sum_{i=1}^m A_i B_i.
    \end{split}
\end{equation*}

Following the steps of polynomial codes, define the polynomials
\begin{equation*}
    \begin{split}
    p_A(x) := \sum_{i=1}^m A_i x^{a_i}
    ,\quad
    p_B(x) := \sum_{i=1}^m B_i x^{b_i},
    \end{split}
\end{equation*}
with the degrees of the monomials satisfying
\begin{equation} \label{equation:matdot_codes}
    \begin{split}
    \exists d \in \mathbb{N} \text{ such that there are exactly } m \text{ pairs } (a_i, b_i) \text{ such that } d = a_i + b_i.
    \end{split}
\end{equation}
If we consider the polynomial $h := p_A p_B$, we observe that, in virtue of (\ref{equation:matdot_codes}), the term of degree $d$ has $AB$ as its coefficient. As in polynomial codes, the polynomials $p_A$, $p_B$ and $h$ ``contain the information'' of $A$,$B$ and $AB$, respectively.

With these polynomials, the design of the algorithm is straightforward by translating that of polynomial codes. We choose distinct points $x_1, x_2, \ldots, x_N \in \mathbb{F}_q$. The master node shares with each one of the worker nodes a point and the polynomials $p_A$ and $p_B$. Then, the worker nodes compute $h(x_i)$ through the identity $h(x_i) = p_A(x_i) p_B(x_i)$. When enough number of them finish, the master node gathers these evaluations and interpolates $h$, recovering $AB$. This could be summarized in a commutative diagram similar to Figure \ref{figure:dmm_framework}.

Following the lines of polynomial codes, the construction of the sets $D_A = \{ a_1, a_2, \ldots, $ $ a_m \}$ and $D_B = \{ b_1, b_2, \ldots, b_m \} $ matters since the interpolation of $h$ is determined by $ \deg(h)+1 = \max(D_A) + \max(D_B)+1$. The following is the solution from \cite{matdot_codes}.

\begin{construction}[\textbf{\cite{matdot_codes}}] \label{construction:optimal_matdot}
    Let $m \in \mathbb{N}^*$. Consider the sets
    \begin{equation*}
        \begin{split}
        D_A = D_B := \{0, 1, \ldots, m-1\}.
        \end{split}
    \end{equation*}
    This construction satisfies (\ref{equation:matdot_codes}). We have that $\max(D_A) + \max(D_B) +1 = 2m - 1$. Indeed, it is ``optimal'' since, for any sets $ D_A^\prime , D_B^\prime  \subseteq \mathbb{N}$ satisfying (\ref{equation:matdot_codes}), they both satisfy $\max(D_A^\prime ) \geq m - 1$ and $ \max(D_B^\prime ) \geq m-1$, hence, $\max(D_A^\prime ) + \max(D_B^\prime ) +1 \geq 2m-1$.
\end{construction}

\subsection{The Algebraic-Geometry solution} \label{subsection:ag_matdot}

As in the case of polynomial codes, matdot codes require the number of workers to satisfy $ N \leq q $, . We proceed as in Section \ref{section:polynomial}, by replacing polynomials with algebraic functions in order to circumvent this issue. Once again, the price to pay in computational complexity and asmptotic parameters is small, see Sections \ref{section:interpolating} and \ref{section:asymptotic}.

Let $F/\mathbb{F}_q$ be an algebraic function field and $Q \in \mathbb{P}$ a place. Let $ \mathcal{B} = \{f_{s_0}, f_{s_1}, \ldots\}$ be a basis of $\mathcal{L}(\infty Q)$ with $\nu(f_{s_i}) = s_i \in S$, the Weierstrass semigroup of $Q$. We start by defining the analogous of Definition \ref{definition:AG_poly_codes_solution} for AG matdot codes:

\begin{definition} \label{definition:AG_matdot_codes}
    Given $m \in \mathbb{N}^*$, we say that $D_A, D_B \subseteq S$ is a solution to the AG matdot code problem if
    \begin{itemize}
        \item The sets $D_A$ and $D_B$ are of size $m$.
        \item There exists $d \in D_A + D_B$ such that there are exactly $m$ pairs $(a, b) \in D_A \times D_B$ satisfying $d = a + b$.
    \end{itemize}
    We define the recovery threshold as $\max(D_A) + \max(D_B)+1$ and we say that the solution is optimal if the recovery threshold is minimum among all the possible solutions (observe that $d$ is not fixed, only $m$ is fixed).
\end{definition}

\begin{remark}
    Observe that if $S = \mathbb{N}$, Definition \ref{definition:AG_matdot_codes} recovers (\ref{equation:matdot_codes}).
\end{remark}

In order to proceed with the DMM algorithm, we need to tweak again the basis $ \mathcal{B} $. The main difficulty is that $f_{a_i} f_{b_j} = f_{a_i + b_j}$ may not always be guaranteed. Theorem \ref{theorem:basis_matdot} below provides the construction of a basis that we will be able to use in general. 

\begin{lemma} \label{lemma:AG_matdot_basis}
    Let $a \in S$ be such that $0 < a < d$. Consider $D \subseteq S$ such that for each $b \in D$ it holds that $a + b \geq d$ and $a > b$. Then, there exists a basis $ \mathcal{B}^\prime  = \{f_{s_1}^\prime , f_{s_2}^\prime , \ldots\}$ satisfying
    \begin{equation*}
        \begin{split}
        \pi_{f_d^\prime }(f_a^\prime  f_b^\prime ) =
            \begin{cases}
                1 & \text{if } a + b = d,\\
                0 & \text{if } a + b \neq d,
            \end{cases}
        \end{split}
    \end{equation*}
    for every $b \in D$.
\end{lemma}

\begin{proof}
    We modify the element $f_a \in \mathcal{B} $ successively to obtain the new basis with the desired property. Consider $D := \{b_1 < b_2 < \ldots < b_r\}$. Define
    \begin{equation*}
        \begin{split}
        f_{a,i} :=
        \begin{cases}
            f_a & \text{if } i = 0,\\
            \pi_{f_d}(f_a f_{b_i})^{-1} f_a & \text{if } i = 1 \text{ and } a + b_1 = d,\\
            f_{a, i-1} - \pi_{f_d}(f_{a,i-1} f_{b_i}) \pi_{f_d}(f_{d-b_i} f_{b_i})^{-1} f_{d - b_i} & \text{otherwise}.
        \end{cases}
        \end{split}
    \end{equation*}
    and consider $f_a^\prime  := f_{a,r}$. Due to the construction of $f_a^\prime $, it is clear that $ \mathcal{B}^\prime  := ( 
\mathcal{B} \setminus{\{f_a\}}) \cup \{f_a^\prime \}$ is a basis of $\mathcal{L} (\infty Q)$.

    Now, let us check that $\mathcal{B}^\prime $ satisfies the property of the statement. If $b_i \in D$ is such that $a + b_i > d$, then $f_a^\prime  = f_{a,i} + g$, for some $g \in \mathcal{L}(\infty Q)$ with $g = 0$ if $i = r$ or $\nu_Q(g) \leq d - b_{i+1}$ otherwise. This is because  $f_a^\prime $ is of the form
    \begin{equation*}
        f_a^\prime  := \lambda_0 f_a - \lambda_1 f_{d - b_1} - \ldots - \lambda_r f_{d - b_r},
    \end{equation*}
    for appropriate $\lambda_0, \lambda_1, \lambda_2, \ldots \lambda_r \in \mathbb{F}_q$. We have that
    \begin{equation} \label{equation:lemma_ag_matdot1}
        \begin{split}
        \pi_{f_d}(f_a^\prime  f_{b_i}) = \pi_{f_d}((f_{a,i} + g) f_{b_i}) = \pi_{f_d}(f_{a,i} f_{b_i}) + \pi_{f_d}(g f_{b_i}).
        \end{split}
    \end{equation}
    Observe that
    \begin{equation} \label{equation:lemma_ag_matdot2}
        \begin{split}
            \pi_{f_d}(f_{a,i} f_{b_i}) &= \pi_{f_d}( (f_{a,i-1} - \pi_{f_d}(f_{a,i-1} f_{b_i}) \pi_{f_d}(f_{d-b_i} f_{b_i}) ^{-1} f_{d - b_i} ) f_{b_i})\\
            &= \pi_{f_d}(f_{a, i-1} f_{b_i}) - \pi_{f_d}(f_{a, i-1} f_{b_i}) \pi_{f_d}(f_{d-b_i} f_{b_i})^{-1} \pi_{f_d}(f_{d-b_i} f_{b_i})\\
            &= 0.
        \end{split}
    \end{equation}
    Morever, if $i \neq r$, the fact $\nu_Q(g) \leq d-b_{i+1} < d - b_{i}$ implies that
    \begin{equation*}
        \begin{split}
        \nu_Q(g f_{b_i}) = \nu_Q(g) + \nu_Q(f_{b_i}) = \nu_Q(g) + b_i < d - b_i + b_i = d,
        \end{split}
    \end{equation*}
    and we deduce that $\pi_{f_d}(g f_{b_i}) = 0$. We obtain the same if $i = r$. Combining this assertion together with (\ref{equation:lemma_ag_matdot1}) and (\ref{equation:lemma_ag_matdot2}) we conclude that
    \begin{equation*}
        \pi_{f_d}(f_a^\prime  f_{b_i}) = 0.
    \end{equation*}
    If $a + b_1 = d$, then $\pi_{f_d}(f_a^\prime  f_b) = 1$ by a similar proof and the lemma follows.
\end{proof}

\begin{theorem} \label{theorem:basis_matdot}
    Let $D_A, D_B \subseteq S$ and $d \in S \setminus (D_A \cup D_B)$. There exists a basis $\mathcal{B}$ of $\mathcal{L}(\infty Q)$ satisfying
    \begin{equation} \label{equation:basis_matdot}
        \begin{split}
            \pi_{f_d} (f_a f_b) =
            \begin{cases}
                1 &\text{if } a + b = d,\\
                0 &\text{if } a + b \neq d,
            \end{cases}
        \end{split}
    \end{equation}
    for every $a \in D_A$ and $b \in D_B$.
\end{theorem}

\begin{proof}
    Given a basis $ \mathcal{B} $, we modify it in ascending order of valuation.

    Define $D := D_A \cup D_B = \{d_1 < d_2 < \ldots < d_r\}$. Let $d_j \in D$ the first element of $D$ such that $2 d_j \geq d$.  If $d_j = \frac{d}{2} \in D_A \cap D_B$, we start by setting $f_d := f_{\frac{d}{2}} f_{\frac{d}{2}}$. Next, for $i = j, j+1, \ldots, r$, in increasing order we apply Lemma \ref{lemma:AG_matdot_basis} for each $d_i \in D$ distinguishing three situations:
    \begin{itemize}
        \item If $d_i \in D_A \setminus{D_B}$, then $f_{d_i} := f_{d_i}^\prime $ with respect to the set
            \begin{equation*}
                \{e \in D_B \such e < d_i,\quad e + d_i \geq d\},
            \end{equation*}
            as in Lemma \ref{lemma:AG_matdot_basis}.
        \item If $d_i \in D_B \setminus{D_A}$, then $f_{d_i} := f_{d_i}^\prime  $ with respect to the set
            \begin{equation*}
                \{e \in D_A \such e < d_i,\quad e + d_i \geq d\},
            \end{equation*}
            as in Lemma \ref{lemma:AG_matdot_basis}.
        \item If $d_i \in D_A \cap D_B$, then $f_{d_i} := f_{d_i}^\prime  $ with respect to the set
            \begin{equation*}
                \{e \in D \such e < d_i,\quad e + d_i \geq d\},
            \end{equation*}
            as in Lemma \ref{lemma:AG_matdot_basis}, and if $2 d_i > d $, we redefine $f_{2 d_i} := f_{d_i}^\prime  f_{d_i}^\prime $.
    \end{itemize}
    Observe that we are modifying recursively each $f_{a_i}$ in such a way that the statement is satisfied for every product $f_{a_i} f_{b_j}$ for $a_i \in D_A$ and $b_j \in D_B$. Since each element of the basis is modified without disturbing this property for elements of lower valuation, the result follows.
\end{proof}

The proof of Theorem \ref{theorem:basis_matdot} gives an algorithm for constructing the basis satisfying (\ref{equation:basis_matdot}). With that basis, we define the functions $p_A \in \mathcal{L}(\infty Q)^{r \times \frac{s}{m}}$ and $p_B \in \mathcal{L}(\infty Q)^{\frac{s}{m} \times t}$ as 
\begin{equation*}
    \begin{split}
    p_A := \sum_{i=1}^m A_i f_{a_i}
    ,\quad
    p_B := \sum_{i=1}^m B_i f_{b_i},
    \end{split}
\end{equation*}
and we define $h := p_A p_B$.  Let $P_1, P_2, \ldots P_N \in \mathbb{P}$ be distinct rational places, we proceed as in matdot codes. The master node shares a place $P_i$ and both $p_A$ and $p_B$ with the $i$th worker node. When enough of them end the computation of $h(P_i) = p_A (P_i) p_B(P_i)$, the master node interpolates $h$, recovering $AB$ as the projection $\pi_{f_d}(h)$ (this is because the basis satisfies (\ref{equation:basis_matdot})).

We dedicate the remainder of the section to finding optimal solutions to the AG matdot code problem.

\begin{construction} \label{construction:AG_matdot_trivial}
    Let $m \in \mathbb{N}^*$ and a numerical semigroup $S$. Consider the sets
    \begin{equation*}
        \begin{split}
        D_A = D_B :=  c(S) + \{0,1,\ldots, m-1\} = \{c(S), c(S) + 1, \ldots, c(S) + m - 1\}.
        \end{split}
    \end{equation*}
    These sets form a solution to the AG matdot code problem, where $d = 2c(S) + m - 1$. Its recovery threshold is also $2(c(S) + m) - 1$.
\end{construction}

Construction \ref{construction:AG_matdot_trivial} is the trivial generalization of Construction \ref{construction:optimal_matdot}. 

The presence of number $d$ in Definition \ref{definition:AG_matdot_codes} gives us some room to work and allows us to exploit the semigroup structure of $S$. In order to study the optimal solutions to the AG matdot code problem, we need some technical results. See \cite{numerical_semigroups}.

\begin{definition} \label{definition:s_order}
    We define the partial order $\leq_S$ over $\mathbb{N}$ given by
    \begin{equation*}
        a \leq_S b \iff b - a \in S.
    \end{equation*}
\end{definition}

\begin{lemma} \label{lemma:AG_matdot_optimal1}
    Let $D_A$ and $D_B$ be an optimal solution to the AG matdot code problem. If $s \in S \cap [\min(D_A), \max(D_A)]$, then
    \begin{equation*}
        \begin{split}
        s \leq_S d \iff s \in D_A.
        \end{split}
    \end{equation*}
\end{lemma}

\begin{proof}
    Obviously, if $s \in D_A$, then $s \leq_S d$. Let us suppose that $s \leq_S d$ and $s \notin D_A$. Then we can define the following solution
    \begin{equation*}
        \begin{split}
        &D_A^\prime  := (D_A \setminus{\max (D_A)}) \cup \{s\},\\
        &D_B^\prime  := d - D_A^\prime ,
        \end{split}
    \end{equation*}
    which satisfies
    \begin{equation*}
        \begin{split}
            \max(D_A^\prime ) + \max(D_B^\prime ) &= d + \max(D_A^\prime ) - \min(D_A^\prime ) \\
            &< d + \max(D_A) - \min(D_A) \\
            &= \max(D_A) + \max(D_B),
        \end{split}
    \end{equation*}
    contradicting the optimality of $D_A$ and $D_B$.
\end{proof}

\begin{lemma} \label{lemma:AG_matdot_optimal2}
    Let $D_A$ and $D_B$ be an optimal solution. Then
    \begin{enumerate}
        \item $\min(D_A) \leq c(S)$.
        \item $\max(D_A) \geq d - c(S)$.
    \end{enumerate}
\end{lemma}

\begin{proof}
    We prove them separately:
    \begin{enumerate}
        \item Suppose that $\min(D_A) > c(S)$. Then we can define the solution
            \begin{equation*}
                \begin{split}
                &D_A^\prime  := D_A - 1,\\
                &D_B^\prime  := d - 1 - D_A^\prime  = D_B,
                \end{split}
            \end{equation*}
            which satifies
            \begin{equation*}
                \begin{split}
                    \max(D_A^\prime ) + \max(D_B^\prime ) &= d - 1 + \max(D_A^\prime ) - \min(D_A^\prime ) \\
                    &= d - 1 + \max(D_A) - \min(D_A)\\
                    &< d + \max(D_A) - \min(D_A)\\
                    &= \max(D_A) + \max(D_B),
                \end{split}
            \end{equation*}
            contradicting the optimality of $D_A$ and $D_B$.

        \item If $\max(D_A) < d - c(S)$, then $\min(D_B) > c(S)$. Take the solution
            \begin{equation*}
                \begin{split}
                &D_A^\prime  := D_B,\\
                &D_B^\prime  := D_A,
                \end{split}
            \end{equation*}
            and the result follows from the first part of the lemma.
    \end{enumerate}
\end{proof}

\begin{definition}
    Let $\delta \in [0, c(S)] \cap S$. Define
    \begin{equation*}
        \begin{split}
        n(\delta) := |[\delta, c(S) - 1] \cap S|,
        \end{split}
    \end{equation*}
    i.e., the number of nontrivial elements of the semigroup greater than or equal to $\delta$.
\end{definition}

We now are ready to prove the main result of this section. Theorem \ref{theorem:AG_matdot_optimal} below gives an explicit description of optimal solutions to the AG matdot code problem with some minor restriction on the size of $m$.

\begin{theorem} \label{theorem:AG_matdot_optimal}
    Let $m \geq 2 c(S)$. Consider an element $\delta \in [0, c(S)] \cap S$ that maximizes $\delta + 2 n(\delta)$. Define $d := m - 1 + 2 c(S) - 2 n(\delta)$. Then the solution
    \begin{equation*}
        \begin{split}
        D_A =  D_B := ([\delta, c(S) - 1] \cap S) \cup ([c(S), d - c(S)]) \cup (d - [\delta, c(S) - 1] \cap S),
        \end{split}
    \end{equation*}
    is optimal, and with recovery threshold $ 2(d -\delta)+1
 $.
\end{theorem}

\begin{proof}
    Consider an optimal solution $D_A^\prime $ and $D_B^\prime $ with $m \geq 2 c(S)$ for $d^\prime $. We can partition $D_A^\prime $ as
    \begin{equation*}
        \begin{split}
        D_A^\prime  = (D_A^\prime  \cap [0,c(S)-1]) \cup (D_A^\prime  \cap [c(S), d^\prime  - c(S)]) \cup (D_A^\prime  \cap [d^\prime  - c(S) + 1, d^\prime ]).
        \end{split}
    \end{equation*}
    Observe the following facts:
    \begin{enumerate}
        \item If $s \in S \cap [0, c(S) - 1]$, then $s \leq_S d^\prime  $ since
            \begin{equation*}
                \begin{split}
                d^\prime  - s \geq m - 1 - s \geq m - 1 - c(S) + 1 = m - c(S) \geq c(S).
                \end{split}
            \end{equation*}
        \item If $s \in [c(S), d^\prime  - c(S)]$, then $s \leq_S d^\prime  $ since
            \begin{equation*}
                \begin{split}
                d^\prime  - s \geq d^\prime  - d^\prime  + c(S) = c(S).
                \end{split}
            \end{equation*}
        \item If $s \in [d^\prime  - c(S)+1, d^\prime ]$, then $s \leq_S d^\prime $ if and only if $d^\prime  - s \in S$ (this is just the definition of $\leq_S$).
    \end{enumerate}
    Applying Lemma \ref{lemma:AG_matdot_optimal1} and Lemma \ref{lemma:AG_matdot_optimal2}, we conclude that $D_A^\prime $ has the form
    \begin{equation*}
        \begin{split}
            D_A^\prime  := ([\min (D_A^\prime ), c(S) - 1] \cap S) \cup [c(S), d^\prime  - c(S)] \cup (d^\prime  - ([d^\prime  - \max(D_A^\prime ), c(S) - 1] \cap S)).
        \end{split}
    \end{equation*}
    Let us define
    \begin{equation*}
        \begin{split}
        &\delta_1 := \min(D_A^\prime ),\\
        &\delta_2 := d^\prime  - \max(D_A^\prime ).
        \end{split}
    \end{equation*}
    Now we can count elements to deduce that
    \begin{equation*}
        \begin{split}
            m = |D_A^\prime | = n(\delta_1) + (d^\prime  - c(S) - c(S) + 1) + n(\delta_2) = n(\delta_1) + n(\delta_2) + d^\prime  - 2c(S) + 1,
        \end{split}
    \end{equation*}
    and, hence $d^\prime  = m - 1 + 2c(S) - n(\delta_1) - n(\delta_2)$. Therefore
    \begin{equation*}
        \begin{split}
            \max(D_A^\prime ) + \max(D_B^\prime ) &= d^\prime  + \max(D_A^\prime ) - \min(D_A^\prime )\\
            &= d^\prime  + d^\prime  - \delta_2 - \delta_1\\
            &= 2(m - 1 + 2c(S)) - (\delta_1 + \delta_2 + 2n(\delta_1) + 2n(\delta_2)).
        \end{split}
    \end{equation*}

    Now consider the sets of the statement, $D_A$ and $D_B$. Clearly this is a solution. Also,
    \begin{equation*}
        \max(D_A) + \max(D_B) = 2(d - \delta) = 2(m - 1 + 2c(S) - 2n(\delta) - \delta).
    \end{equation*}
    Because of the definition of $\delta$, we have that $\max(D_A) + \max(D_B) \leq \max(D_A^\prime ) + \max(D_A^\prime )$, concluding that $D_A$ and $D_B$ is an optimal solution.
\end{proof}

\begin{remark}
    Observe that the number $\delta$ defined in Proposition \ref{theorem:AG_matdot_optimal} is independent of $m$ (as long as $m \geq 2 c(S)$), so we only need to compute it once for the chosen semigroup. It is an invariant of the semigroup as $g(S)$, $n(S)$ and $m(S)$ are.
\end{remark}

\begin{remark}
    Observe that if $S = \mathbb{N}$, the solution defined in Theorem \ref{theorem:AG_matdot_optimal} is $D_A = D_B = [0, m-1]$, the same of Construction \ref{construction:optimal_matdot}.
\end{remark}

\begin{definition}
    Define the map $\Delta$ given by
    \begin{equation*}
        \begin{split}
        \Delta: S \cap [0, c(S)] &\rightarrow \mathbb{N}\\
        \delta &\mapsto \delta + 2 n(\delta).
        \end{split}
    \end{equation*}
\end{definition}

Because of Proposition \ref{theorem:AG_matdot_optimal}, it makes sense to consider $\Delta$, since studying its maximums leads to finding optimal solutions. Lemma \ref{lemma:Delta_properties} summarizes some properties of $\Delta$:

\begin{lemma} \label{lemma:Delta_properties}
    $\Delta$ satisfies the following:
    \begin{enumerate}
        \item $\Delta(0) = 2 n(S)$.
        \item $\Delta(c(S)) = c(S)$.
        \item If $s, s^\prime  \in S \cap [0, c(S)]$ are such that $s = \max \{t \in S \such t < s^\prime  \}$, then:
            \begin{equation*}
                \begin{split}
                    s^\prime - s = 1      &\implies \Delta(s^\prime) < \Delta(s),\\
                    s^\prime - s = 2      &\implies \Delta(s^\prime) = \Delta(s),\\
                    s^\prime - s \geq 3   &\implies \Delta(s^\prime) > \Delta(s).
                \end{split}
            \end{equation*}
    \end{enumerate}
\end{lemma}

\begin{proof}
    Trivial by the definition of $\Delta$.
\end{proof}

Even though the domain of $\Delta$ is finite and we can theoretically use brute force to find its maximum, we bound in Proposition \ref{proposition:Delta_bound} the set where it can be found. We need the $\phi$ map, which can be found implicitly defined in \cite{froberg}, for example.

\begin{definition}
    Let $S$ be a numerical semigroup. We define the map $\phi$ as
    \begin{equation*}
        \begin{split}
            \phi : [0, c(S) - 1] &\to [0, c(S) - 1]\\
            x &\mapsto c(S) - 1 - x
        \end{split}
    \end{equation*}
\end{definition}

\begin{lemma}
    Let $x, y \in [0, c(S) - 1]$. Then $\phi$ satisfies the following:
    \begin{itemize}
        \item $\phi$ is involutive, i.e., $\phi^2 = \id$.
        \item $x \leq_S y$ if and only if $\phi(x) \leq_S \phi(y)$.
        \item If $x \in S$, then $\phi(x) \notin S$. The converse is true if and only if $S$ is symmetric i.e. $g(S) = n(S)$.
    \end{itemize}
\end{lemma}

\begin{proof}
    Trivial by the definition of $\phi$.
\end{proof}

\begin{proposition} \label{proposition:Delta_bound}
    The map $\Delta$ reaches its maximum in some $\delta \geq c(S)/2$.
\end{proposition}

\begin{proof}
    Let us suppose that $\Delta$ reaches its maximum in $\delta < c(S)/2$. Consider
    \begin{equation*}
        \begin{split}
        \delta^\prime := \min \{ s \in S \such s > \phi(\delta) \},
        \end{split}
    \end{equation*}
    and the number
    \begin{equation*}
        \begin{split}
        \Delta(\delta^\prime) - \Delta(\delta) = (\delta^\prime - \delta) - 2 (n(\delta) - n(\delta^\prime)).
        \end{split}
    \end{equation*}
    First,
    \begin{equation*}
        \begin{split}
        \delta^\prime - \delta \geq \phi(\delta) + 1 - \delta = c(S) - 2 \delta.
        \end{split}
    \end{equation*}
    Second,
    \begin{equation*}
        \begin{split}
        n(\delta) - n(\delta^\prime) = |S \cap [\delta, \delta^\prime - 1]| = |S \cap [\delta, \phi(\delta)]| + |S \cap [\phi(\delta) + 1, \delta^\prime - 1]|.
        \end{split}
    \end{equation*}
    Due to the construction of $\delta^\prime$, the set $S \cap [\phi(\delta)+1, \delta^\prime - 1]$ is empty, so
    \begin{equation*}
        \begin{split}
        n(\delta) + n(\delta^\prime) = |S \cap [\delta, \phi(\delta)]| \leq \frac{\phi(\delta) - \delta + 1}{2} = \frac{c(S)}{2} - \delta.
        \end{split}
    \end{equation*}
    Using these two inequalities we get
    \begin{equation*}
        \begin{split}
        \Delta(\delta^\prime) - \Delta(\delta) \geq c(S) - 2\delta - 2\left( \frac{c(S)}{2} - \delta \right) = 0,
        \end{split}
    \end{equation*}
    concluding that $\Delta(\delta^\prime) \geq \Delta(\delta)$, so $\Delta$ reaches its maximum in $\delta^\prime > c(S)/2$.
\end{proof}

The remainder of the section is dedicated to studying specific families of semigroups and computing where $\Delta$ reaches its maximum.

\begin{definition}
    We say that a numerical semigroup is sparse if it has no consecutive elements lower than the conductor.
\end{definition}

    As an example of sparse semigroups highly related to coding theory, we can consider the semigroups associated to the second tower of function fields of Garc{\'i}a-Stichtenoth \cite{garcia_stichtenoth_semigroup}. For more information about sparse semigroups we refer to \cite{sparse_numerical_semigroups}. Computing the maximum of the $\Delta$ map is easy for sparse semigroups as Proposition \ref{proposition:Delta_sparse} shows.

\begin{proposition} \label{proposition:Delta_sparse}
    If $S$ is sparse, then $\Delta$ reaches its maximum at $\delta = c(S)$.
\end{proposition}

\begin{proof}
    If $S$ is sparse, the last statement of Lemma \ref{lemma:Delta_properties} implies that $\Delta$ is nondecreasing when restricted to $S$ and the result follows.
\end{proof}

The other family of semigroups we will study in this paper is the following.

\begin{definition} \label{definition:hermitian_semigroups}
    A numerical semigroup is Hermitian if $S = \langle q, q+1 \rangle$ with $q > 1$.
\end{definition}

\begin{remark} \label{remark:hermitian}
    Some observations about Hermitian semigroups:
    \begin{enumerate}
        \item Since they are generated by two elements, we have that $c(S) = q(q-1)$ and $n(S) = g(S) = c(S)/2 = q(q-1)/2$. In particular, they are symmetric.
        \item The elements of the semigroup smaller than the conductor are exactly
            \begin{equation*}
                \begin{split}
                S \cap [0, c(S) - 1] = \{i + kq \in \mathbb{N} \such k \in [0, q-2], \quad i \in [0, k]\}.
                \end{split}
            \end{equation*}
    \end{enumerate}
\end{remark}

When $q$ is a prime power, Hermitian semigroups are Weierstrass semigroups  associated to Hermitian curves \cite[Section~5.3]{pellikaan}. As with sparse semigroups, the $\Delta$ map can be easily studied over Hermitian semigroups.

\begin{proposition}
    If $S = \langle q, q+1 \rangle$ is Hermitian, then $\Delta$ reaches its maximum at $\delta = q \lceil (q - 1)/2 \rceil$.
\end{proposition}

\begin{proof}
    First, observe that $\delta$ must be $\delta = k q$ for some $k \in [0, q-1]$, since, if $\delta - 1 \in S$, the last statement of Lemma \ref{lemma:Delta_properties} implies that $\Delta(\delta) < \Delta(\delta - 1)$. Because of the last statement of Remark \ref{remark:hermitian}, we know the elements of the semigroup under the conductor, so we can easily give an explicit formula for $n(kq)$:
    \begin{equation*}
        \begin{split}
            n(kq) &= \sum_{i=k}^{q-2} (i + 1)\\
            &= (q - 1 - k) + \sum_{i=k}^{q-2} i\\
            &= (q - 1 - k) + k(q - 1 - k) + \sum_{i=0}^{q-2-k} i\\
            &= (q - 1 - k) + k(q - 1 - k) + \frac{(q - 2 - k) (q - 1 - k)}{2}\\
            &= (q - 1 - k)\left(\frac{q+k}{2}\right).
        \end{split}
    \end{equation*}
    The next step is to maximize the function $f(k) := \Delta(kq)$:
    \begin{equation*}
        \begin{split}
        f(k) = kq + 2 (q - 1 - k)\left(\frac{q+k}{2}\right) = -k^2 + (q - 1)k + q(q - 1),
        \end{split}
    \end{equation*}
    which reaches its maximum in $k = \lceil(q-1)/2\rceil$. We conclude that $\Delta$ reaches its maximum in $\delta = q \lceil (q - 1)/2 \rceil$.
\end{proof}

\section{Decoding and complexity} \label{section:interpolating}

We now complete the procedures from Subsections \ref{subsection:ag_polynomial} and \ref{subsection:ag_matdot} by showing how to recover $h \in \mathcal{L}(kQ)^{a \times b}$ from its evaluations in $k+1$ $ \mathbb{F}_q $-rational places $P_1, P_2, \ldots, P_{k+1} \in \mathbb{P}$ (these evaluations are the computations of the first $k$ responsive nodes), where $ k = \max(D_A+D_B) $. Observe that $ (a,b) = (r/m,t/n) $ for AG polynomial codes and $ (a,b) = (r,t) $ for AG matdot codes. More precisely, we need to recover the coordinates of $h$ with respect to the basis $\mathcal{B}$, defined in Subsection \ref{subsection:ag_polynomial} for AG polynomial codes and in Theorem \ref{theorem:basis_matdot} for AG matdot codes. This can be achieved using Lagrange interpolation as follows.

For some $(i,j) \in [1, a] \times [1,b]$, we write $h_{i,j} \in \mathcal{L}(k Q)$ for the $(i,j)$th entry of $h$. Let $ \kappa \in \mathbb{N} $ be such that $ [0,k]\cap S = \{s_0,s_1, \ldots, s_\kappa \} $ and consider the linear equations over $ \mathbb{F}_q $
\begin{equation} \label{equation:interpolation}
    \begin{cases}
        x_0 f_{s_0}(P_1) + x_1 f_{s_1}(P_1) + \cdots + x_\kappa f_{s_\kappa}(P_1) & = h_{i,j}(P_1), \\
        x_0 f_{s_0}(P_2) + x_1 f_{s_1}(P_2) + \cdots + x_\kappa f_{s_\kappa}(P_2) & = h_{i,j}(P_2), \\
                                                                                & \vdots         \\
        x_0 f_{s_0}(P_{k+1}) + x_1 f_{s_1}(P_{k+1}) + \cdots + x_\kappa f_{s_\kappa}(P_{k+1}) & = h_{i,j}(P_{k+1}),
    \end{cases}
\end{equation}
whose solutions $x_0, x_1, \ldots x_\kappa \in \mathbb{F}_q$ form the coordinates of $h_{i,j}$ with respect to the basis $\mathcal{B}$. Consider the analogous of the Vandermonde matrix:
\begin{equation*}
    G :=
    \begin{pmatrix}
        f_{s_0}(P_1)     & f_{s_0}(P_2)     & \ldots & f_{s_0}(P_{k+1}) \\
        f_{s_1}(P_1)     & f_{s_1}(P_2)     & \ldots & f_{s_1}(P_{k+1}) \\
        \vdots           & \vdots           & \ddots & \vdots       \\
        f_{s_\kappa}(P_1) & f_{s_\kappa}(P_2) & \ldots & f_{s_\kappa}(P_{k+1})
    \end{pmatrix}.
\end{equation*}
Using this matrix, system (\ref{equation:interpolation}) can be written as
\begin{equation} \label{equation:decoding_coordinate}
    (x_0, x_1, \ldots, x_\kappa) G = (h_{i,j}(P_1), h_{i,j}(P_2), \ldots, h_{i,j}(P_{k+1})).
\end{equation}
Observe that $G$ is the matrix associated to $ev$, the evaluation map from Section \ref{section:preliminaries}. Since $ev$ is injective, $G$ has a right inverse $ G^{-1} $, which can be computed with complexity $\mathcal{O}(k^3)$ using Gaussian elimination. Hence $x_l$ can be obtained as
\begin{equation*}
    x_l = (h_{i,j}(P_1), h_{i,j}(P_2), \ldots, h_{i,j}(P_{k+1})) \, \text{col}_l(G^{-1})
\end{equation*}
with complexity $\mathcal{O}(k)$. Observe that we only need to compute $G^{-1}$ once since this matrix is the same for every index $(i,j)$.

In AG polynomial codes (Section \ref{section:polynomial}), all $\kappa+1$ coordinates of $h$ have to be recovered. So we have to perform $ \frac{rt}{mn} $ componentwise interpolations, each one having a total cost of $\mathcal{O}(k^2)$ plus the complexity of inverting $G$. Then, the total cost of decoding AG polynomial codes is $\mathcal{O}(\frac{rt}{mn} k^2 + k^3)$. In AG matdot codes (Section \ref{section:matdot}), only the coordinate corresponding to $f_d$ has to be obtained. That is, for each $(i,j)$, we have to obtain just the value of $x_d$. So we need to perform $rt$ vector-matrix multiplications of cost $\mathcal{O}(k)$ and compute $G^{-1}$, resulting in a total cost of $\mathcal{O}(rtk + k^3)$. 

Another aspect to keep in mind is the complexity of the operations carried out by each worker node. In AG polynomial codes, the $i$th node has to multiply $p_A(P_i) \in \mathbb{F}_q^{\frac{r}{m} \times s}$ by $p_B(P_i) \in \mathbb{F}_q^{s \times \frac{t}{n}}$. By using the naive multiplication algorithm, this has a cost of $\mathcal{O}(\frac{r s t}{mn})$. In AG matdot codes, the multiplication is $p_A(P_i) \in \mathbb{F}_q^{r \times \frac{s}{m}}$ by $p_B(P_i) \in \mathbb{F}_q^{\frac{s}{m} \times t}$, which has a cost of $\mathcal{O}(\frac{rst}{m})$. Table \ref{table:complexity} summarizes the complexity of AG polynomial codes and AG matdot codes.

\begin{table}[h]
    \centering
    \begin{tabular}{l|c|c}
                            & Worker computation              & Decoding computation \\ \hline
        AG polynomial codes & $\mathcal{O}(\frac{r s t}{mn})$ & $\mathcal{O}(\frac{rt}{mn} k^2 +k^3)$ \\ \hline
        AG matdot codes     & $\mathcal{O}(\frac{r s t}{m})$  & $\mathcal{O}(rtk+k^3)$   \\ \hline
    \end{tabular}
    \caption{Complexity of AG polynomial and matdot codes.}
    \label{table:complexity}
\end{table}

Observe that the decoding complexity is negligible compared to the computation of each worker if $ k^2 = o(\frac{\min\{r,s,t\}}{mn}) $ for AG polynomial codes, and if $ k = o(\frac{\min\{r,s,t\}}{m}) $ for AG matdot codes. Both hypotheses are clearly satisfied in practical scenarios.
%

\section{Asymptotic analysis of the recovery threshold} \label{section:asymptotic}

We conclude by studying the asymptotic performance of the recovery threshold of the AG polynomial codes and the AG matdot codes from Subsections \ref{subsection:ag_polynomial} and \ref{subsection:ag_matdot}. As done in \cite{matdot_codes}, we consider the same storage constraint ($1/m$), that is, the same value of $ m = n $ for AG polynomial codes and AG matdot codes in order to compare them.

We observe that the computational complexity per worker and the decoding complexity, studied in Section \ref{section:interpolating}, behave exactly as in classical polynomial codes \cite{polynomial_codes} and classical matdot codes \cite{matdot_codes}. As in the classical case, AG polynomial codes outperform AG matdot codes in the computational complexity per worker (see Table \ref{table:complexity}). On the other hand, AG matdot codes outperform AG polynomial codes in terms of the recovery threshold, which is around $ m^2 + c(S) $ (quadratic) for the former and $ \leq 2(m+c(S)) $ (linear) for the latter, in our constructions. However, this is if we consider the semigroup constant, which implies that the algebraic function field is kept constant. In particular, the number of workers can only grow until the number of $ \mathbb{F}_q $-rational places of the given algebraic function field. Note that if we change the function field, then $ c(S) $ also changes, and so does the recovery thresholds.
 
We now provide an asymptotic study of the recovery thresholds of both AG polynomial codes and AG matdot codes, when considering sequences of algebraic function fields, and show that they perform almost optimally in the asymptotic regime, for a fixed moderate $ q $. Define the recovery threshold ratio as

%
%
$$ \rho := \frac{\text{recovery threshold}}{N}. $$ 

In the case of AG polynomial codes, the recovery threshold has to be at least $m^2$, hence $ \rho \geq \frac{m^2}{N} $. We obtain an optimal scheme when $ \rho = \frac{m^2}{N} $ (remember that we want the recovery threshold to be as small as possible). Observe that this holds for classical polynomial codes. In the case of AG polynomial codes, the recovery threshold in Constructions \ref{construction:AG_polynomial_Apery} and \ref{construction:AG_polynomial_lemma} is typically $ m^2 + c(S) $. Choose an optimal tower of algebraic function fields, that is, a sequence such that the $i$th function field $F_i / \mathbb{F}_q$ has genus $g(S_i)$ and number of $ \mathbb{F}_q $-rational points $N_i$, where
$$\lim_{i \to \infty} \frac{g(S_i)}{N_i} = \frac{1}{\sqrt{q}-1}, $$ 
being $ q $ a square. Such families exist, for example Garc{\'i}a-Stichtenoth's second tower of function fields \cite{garcia1996asymptotic} (see also \cite[Section~2.9]{pellikaan} and \cite{garcia_stichtenoth_semigroup}), which further satisfies $ \lim_{i \to \infty} g(S_i)/$ $c(S_i) = 1 $, thus
\begin{equation*}
    \lim_{i \to \infty} \left(\rho - \frac{m^2}{N_i}\right) = \lim_{i \to \infty} \frac{c(S_i) + m^2 - m^2}{N_i} = \lim_{i \to \infty} \frac{c(S_i)}{N_i} = \frac{1}{\sqrt{q} - 1}.
\end{equation*}
Hence, we observe that, asymptotically, the difference between the ratios $ \rho - \frac{m^2}{N_i} $ tends to a value $ \varepsilon(q) = (\sqrt{q}-1)^{-1} > 0 $ that is small even for moderate values of $ q $, as was the case of algebraic-geometry codes in classical Coding Theory \cite{pellikaan,stichtenoth}.

The case of AG matdot codes is analogous, taking into account that the recovery threshold must be at least $ 2m-1 $, being classical matdot codes optimal. Using the construction from Theorem \ref{theorem:AG_matdot_optimal} applied to Garc{\'i}a-Stichtenoth's second tower, we obtain the recovery threshold
$$ 2(d-\delta)+1 = 2(m+c(S)) - 1, $$
where $ d = m-1+2c(S)-2n(\delta) $, and $ \delta = c(S) $ (thus $ n(\delta) = 0 $) by Proposition \ref{proposition:Delta_sparse}, since the corresponding semigroups are sparse. Once again, since this tower is optimal and satisfies $ \lim_{i \to \infty} g(S_i)/c(S_i) = 1 $, we have
\begin{equation*}
    \lim_{i \to \infty} \left(\rho - \frac{2m-1}{N_i}\right) = \lim_{i \to \infty} \frac{2(m+c(S_i))-1 - (2m-1)}{N_i} = \lim_{i \to \infty} \frac{2c(S_i)}{N_i} = \frac{2}{\sqrt{q} - 1},
\end{equation*}
asymptotically achieving an almost optimal ratio $ \rho $ even for moderate values of $ q $ (which can be kept fixed as $ i \to \infty $).

\section*{Open problems}
There are two immediate topics to explore related to the contents of this paper:
\begin{itemize}
    \item Improving the constructions of Subsection \ref{subsection:ag_polynomial} and finding optimal solutions to the polynomial code problem in general or when restricting to some families of numerical semigroups.
    \item Generalizing the construction of polydot codes given in \cite{matdot_codes} in the same way as AG polynomial codes and AG matdot codes. Defining the analogous of Definition \ref{definition:AG_poly_codes_solution} and Definition \ref{definition:AG_matdot_codes} is not difficult but the process of tweaking the basis of $\mathcal{L} (\infty Q)$ does not seem trivial.
\end{itemize}


\nocite{*}
\printbibliography

\end{document}